\documentclass{article16}
\usepackage[]{algorithm2e}
\usepackage{filecontents}
\usepackage{mathtools}

\DeclarePairedDelimiter\abs{\lvert}{\rvert}
\DeclarePairedDelimiter\norm{\lVert}{\rVert}
\makeatletter
\let\oldabs\abs
\def\abs{\@ifstar{\oldabs}{\oldabs*}}
\let\oldnorm\norm
\def\norm{\@ifstar{\oldnorm}{\oldnorm*}}
\makeatother

\title{Guarding Path Polygons with Orthogonal Visibility}
\author{Hamid Hoorfar \thanks{Department of Computer Engineering and Information Technology, Amirkabir University of Technology (Tehran Polytechnic), {\tt \{hoorfar,ar\_bagheri\}@aut.ac.ir}}\and Alireza Bagheri\footnotemark[1]~\footnote{\textit{Coresponding author}.}}

\makeatletter
\let\runauthor\@author
\let\runtitle\@title
\makeatother
\begin{document}
\maketitle

\begin{abstract}
We are interested in the problem of guarding simple orthogonal polygons with the minimum number of $ r $-guards. The interior point $ p $ belongs an orthogonal polygon $ P $ is visible from $ r $-guard $ g $, if the minimum area rectangle contained $ p $ and $ q $ lies within $ P $. A set of point guards in polygon $ P $ is named guard set (as denoted $ G $) if the union of visibility areas of these point guards be equal to polygon $ P $ i.e. every point in $ P $ be visible from at least one point guards in $ G $. For an orthogonal polygon, if dual graph of vertical decomposition is a path, it is named path polygon. In this paper, we show that the problem of finding the minimum number of $ r $-guards (minimum guard set) becomes linear-time solvable in orthogonal path polygons. The path polygon may have dent edges in every four orientations. For this class of orthogonal polygon, the problem has been considered by Worman and Keil who described an algorithm running in  $ O(n^{17} poly\log n) $-time where $ n $ is the size of the input polygon. The problem of finding minimum number of guards for simple polygon with general visibility is NP-hard, even if polygon be orthogonal. Our algorithm is purely geometric and presents a new strategy for $ r $-guarding orthogonal polygons and guards can be placed everywhere in the interior and boundary of polygon. 
\end{abstract}

\section{Introduction}
The target of the art gallery problem is finding a set $ G $ of point guards in polygon $ P $ such that every point in $ P $ is visible from some members of $ G $ where a guard $ g $ and a point $ p $ are visible if the line-segment $ gp $ is contained in $ P $. It is shown that finding the optimum number of guards (the minimum guard set) required to cover an arbitrary simple polygon is NP-hard~\cite{lee1986computational}. The art gallery problem is also NP-hard for orthogonal polygons and even remains NP-hard for monotone polygons~\cite{schuchardt1995two}. In the \textit{orthogonal art gallery} problem, it is assumed that the visibility is in orthogonal mode instead of standard line visibility. In the polygon $ P $ and under \textit{orthogonal visibility(r-visibility)}, points $ p $ and $ q $ are visible from each other, if the minimum axis-aligned rectangle spanned by these two points is contained in $ P $, This kind of visibility is , also, called \textit{r-visibility}~\cite{o2004visibility} i.e. two points $ p $ and $ q $ are \textit{r-visible} (\textit{orthogonally visible}) from each other if the minimum area rectangle contained $ p $ and $ q $ has no intersection with the exterior of $ P $. A polygon is \textit{orthogonal} if its edges are either horizontal or vertical, in every orthogonal polygon the number of vertical edges is equal to the number of horizontal ones. Worman and Keil~\cite{worman2007polygon} studied the decomposition of orthogonal polygons into optimum number of r-star(star-shaped) sub-polygons that is equivalent to the orthogonal art gallery problem. They presented a polynomial-time algorithm for the problem under r-visibility, so, they showed that the problem is polynomially solvable. Their algorithm is processable in $ O(n^{17}poly \log n) $, hence, it is not so fast. A \textit{cover} of a polygon $ P $ by a set $ S $ of sub-polygons is defined such that the union of the sub-polygons in $ S $ be equal to $ P $ and the sub-polygons are required to be mutually disjoint except along their boundaries. \textit{$ r $-star} is an orthogonal star-shaped polygon, and every r-star polygons are \textit{orthoconvex} that will defined later. Clearly, the problem of determining a minimum cover of a simple orthogonal polygon by r-stars is equivalent to determining a minimum set of r-visibility guards to guard the entire polygon i.e. finding minimum covers by star-shaped sub-polygons is equivalent to finding the minimum guard set needed such that every point in the polygon is visible to some guards. A linear-time ($O(n)$-time) algorithm for covering a $ x $-monotone orthogonal polygon with the minimum number of $ r $-star polygons was presented by Gewali and et. al. ~\cite{gewali1996placing}. Palios and Tzimas~\cite{palios2014minimum} considered the problem on class-3 orthogonal polygons without holes, i.e., orthogonal polygons that have reflex edges (dents) along at most $ 3 $ different orientations. They presented an algorithm with time complexity of $ O(n+k \log ⁡k)$ where $ k $ is the size of a minimum r-star cover(the size of output). It is shown that problem is NP-hard on orthogonal polygons with holes by Beidl and Mehrabi~\cite{biedl2016r}. A polygon is named \textit{tree polygon} if dual graph of the polygon is an undirected graph in which any two nodes are connected by exactly one path(tree graph). Also, They gave an algorithm for tree polygon in $ O(n) $-time. If vertex guards are only allowed (vertex guard variant), iterations of their algorithm yields an $ O(n^4) $ solution for general orthogonal polygons~\cite{couto2007exact}. If $ S $ be a set of points in the polygon $ P $, and every two points of $ S $ are not visible from each other, then $ S $ is called textit{hidden set}. So, If a hidden set is also a guard set, it is called \textit{hidden guard set}. Hoorfar and Bagheri~\cite{hoorfar2017minimum} showed that finding the minimum number of guards is linear-time even under the constraint that the guards are hidden from each other, for some monotone polygons in the orthogonal art gallery problem. In this paper, we study the orthogonal art gallery problem on path orthogonal polygons. We take advantage of geometric properties of these polygons and we present an $ 1 $-pass $ O(n) $-time algorithm to report the locations of a minimum-cardinality set of $ r $-visibility guards to cover the entire polygon, where $ n $ is the number of vertices of given path polygon. This is the one of the few purely geometric algorithm for this problem. The first one is presented by Palios and Tzimas~\cite{palios2014minimum} and another one is given by Hoorfar and Bagheri~\cite{hoorfar2017minimum,hoorfar2017linear}. Actually, we generalize the ideas of the latter papers to yield faster algorithms for the problem on path orthogonal polygons. Note that a path polygon is not tree polygon and have the dent edges in every four orientations, so the fastest known algorithm for it, have time complexity of $ O(n^{17}poly \log n) $ and presented by Worman and Keil~\cite{worman2007polygon}. In the other word, we show that the $ r $-guarding problem is linear-time solvable on path polygons without holes. Comparing our results to the one by Worman and Keil, their algorithm works for a broader class of polygons, but is too slower. In this paper, visibility means r-visibility (orthogonal visibility) and guarding is under r-visibility and also, monotonicity means $ x $-monotonicity unless explicitly mentioned.

\section{Preliminaries}
\label{ss:ss1}
Assume $ P $ is an orthogonal polygon with $ n $ edges, the interior angles of all reflex vertices belonged to $ P $ are equal to $ \frac{3\pi}{2} $ and  the interior angles of all convex vertices belonged to $ P $ are equal to $ \frac{\pi}{2} $. It is obvious that the number of reflex vertices of an orthogonal polygon with $ n $ vertices is equal to $ \frac{n-4}{2} $ and the number of its convex vertices is equal to $ \frac{n+4}{2} $, exactly. A \textit{decomposition} of an orthogonal polygon $ P $ obtain by extending the edges of $ P $ incident to their reflex vertices until intersect the boundary. Therefore, by plotting these vertical and horizontal line segments, at most $ (\frac{n-2}{2})^2 $ rectangles are obtained, then we have a \textit{partition}, such that the union of the parts of partition be equal to $ P $ and the parts be mutually disjoint except along their boundaries. Every obtained rectangle part is named $ pixel $. If we assign a node to each pixel and then connect every two nodes of which their corresponding pixels are adjacent, by one edge, the created graph is called \textit{dual graph}. Some orthogonal polygons are named according to the type of their dual graphs. An orthogonal polygon is called \textit{tree}, if dual graph of the polygon is a tree and an orthogonal polygon is called \textit{$ k $-width} where its dual graph be a $ k $-width tree. If after decomposition of $ P $, the vertices of all the pixels lie on the boundary of $ P $, the polygon is named \textit{thin}. A tree polygon is a thin polygon without hole. A \textit{vertical decomposition} of an orthogonal polygon $ P $  with $ n $ vertices obtain by extending only the vertical edges of $ P $ incident to their reflex vertices until intersect the boundary. So, after the vertical decomposition of $ P $, at most $ \frac{n-2}{2} $ rectangles will be obtained, this kind of partition is called \textit{vertical partition}. If we assign a node to each rectangle and then connect every two nodes of which their corresponding rectangles are adjacent, by one edge, the created graph is called dual graph of vertical decomposition. An orthogonal polygon is called \textit{path }, if dual graph of its vertical decomposition (not general decomposition) is a path, see figure~\ref{fi:fig1}.
\begin{figure}
	\centering
	\includegraphics[width=\textwidth]{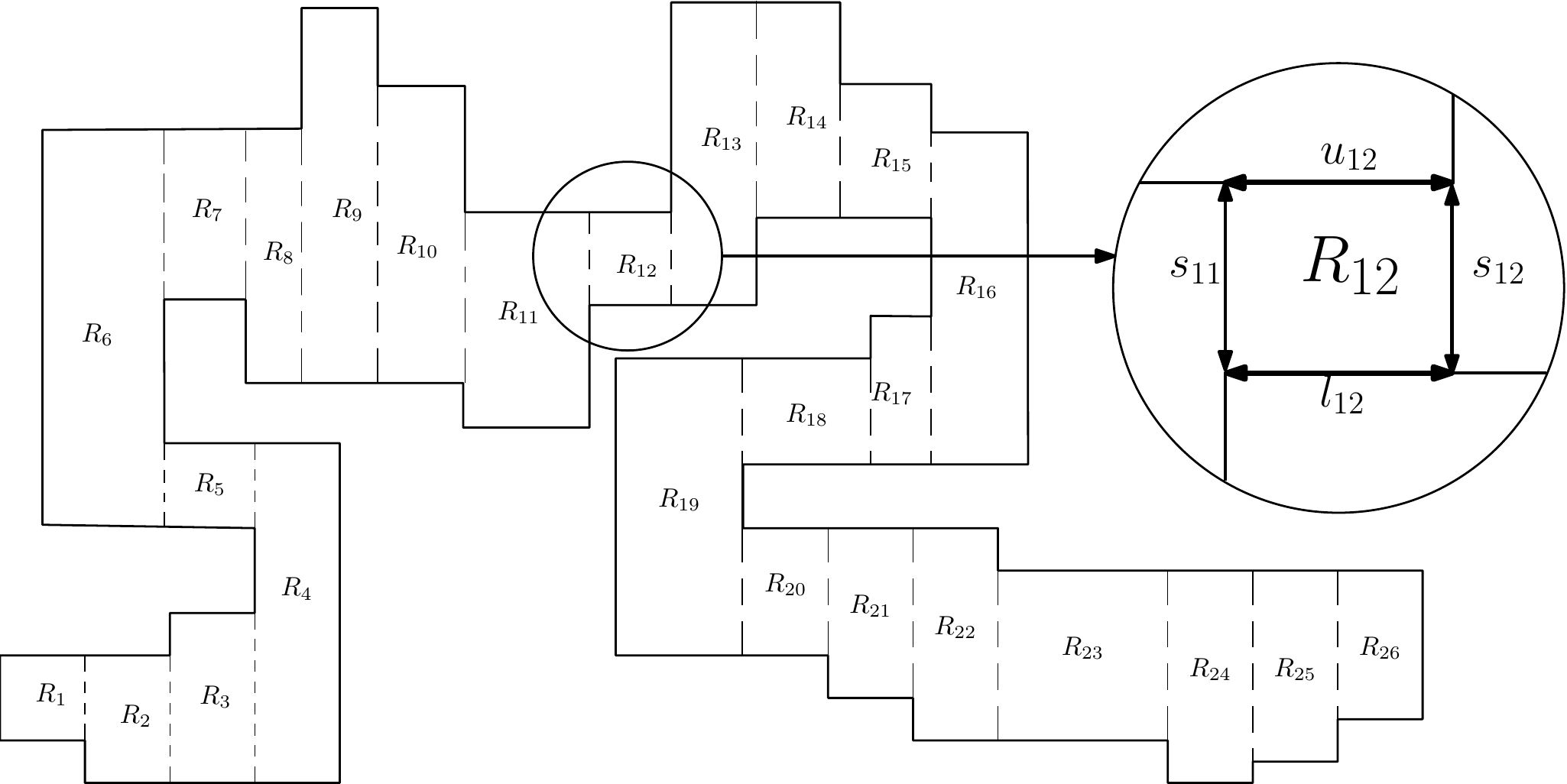}
	\caption{Illumination of the vertical decomposition of a path polygon and its notations.}
	\label{fi:fig1}
\end{figure}
The number of horizontal and vertical edges of an orthogonal polygon is the same. If an edge $ \epsilon_1\in P $ has two endpoints of angle $ \frac{\pi}{2} $, it is called \textit{tooth} edge and if an edge $ \epsilon\in E $ has two endpoints of angle $ \frac{3\pi}{2} $, it is called \textit{dent} edge. The \textit{edge direction} is defined as same as the direction of its normal vector from interior to exterior of the polygon. For $ i=1,2,3,4 $, the \textit{class-$ i $} of orthogonal polygon contains the polygons which have dent edge only in $ i $ different directions~\cite{palios2014minimum}. Every orthogonal polygon has tooth edges in all four directions, every orthogonal $ x $-monotone polygon has no dent edge in the directions perpendicular to the $ y $-axis and every orthogonal $ y $-monotone polygon has no dent edge in the directions perpendicular to the $ x $-axis. Not every polygon has dent edge, hence, an orthogonal polygon that has no dent edge is named \textit{orthogonally convex} and sometimes it is also called \textit{orthoconvex} polygon. The orthoconvex polygon is both $ x $-monotone and $ y $-monotone i.e. if polygon $ P $ is $ x $-monotone and also $ y $-monotone, then $ P $ is orthoconvex. Assume we decompose a simple path polygon $ P $ with $ n $ edges into rectangular parts (rectangles) obtained by extending every vertical edges incident to their reflex vertices of $ P $. The dual graph $ G $ of this vertical decomposition is path which has two node of degree one. The rectangles corresponding to these two nodes are called \textit{first rectangle} and \textit{last rectangle}. Let $R =\{R_{1},R_{2},\dots,R_{m}\} $ be the set of rectangles, where $ m = \frac{n-2}{2} $, ordered from first to last rectangles according to the order of their corresponding nodes in the graph $ G $. For an illustration see figure~\ref{fi:fig1}, $ R_1 $ and $ R_{26} $ are first and last rectangle in this example, respectively. We denote the upper horizontal edges of rectangle $ R_{i} $ by $ u_{i} $ and the lower horizontal edges of $ R_i $ by $ l_{i} $. Let consider that $  U =\{u_{1},u_{2},u_3,\dots,u_{m} \}  $ and $ L=\{l_{1},l_{2},l_3,\dots,l_{m} \} $. Two consecutive rectangles are mutually disjoint except along their boundaries, hence, the intersection between every two consecutive rectangles $ R_i $ and $ R_{i+1} $ is a vertical segment that is denoted as $ s_i $. See the rectangle that is enclosed in a circle in figure~\ref{fi:fig1} for the illumination. For every horizontal segment $ s $ the $ y $-coordinate of every points on $ s $ is the same, so, we denote the $ y $-coordinate of $ s $ by $ y(s) $. Similarly, For every vertical segment $ s' $ the $ x $-coordinate of  every points on $ s $ is the same, hence, we denote the $ x $-coordinate of $ s' $ by $ x(s') $. For a point $ p $, $ y $-coordinate and $ x $-coordinate of $ p $ is denoted by $ y(p) $ and $ x(p) $, respectively. Without reducing generality, We assume that for every two different vertical edges $ e $ and $ e' $, the $ x $-coordinates of both of them is not same($ x(e)\neq x(e') $), hence, it is clear that for every {\small $ 1 \leq j \leq m-1 $}, $ y(u_{j})=y(u_{j+1}) $ or $ y(l_{j})=y(l_{j+1}) $. Also, we denote the horizontal edge of $ P $ that contains segment $ u_{k} $  by $ e(u_{k}) $ and the horizontal edge of $ P $ that contains segment $ l_{i} $  by $ e(l_{i}) $. Let the sets $ E_{U}=\{e(u_{j})|1 \leq j \leq m\} $ and $ E_{L}=\{e(l_{j})|1 \leq j \leq m\} $ be sets of horizontal edges of upper chain and lower chain of $ P $ ordered corresponding to their rectangles order. In the set $ E $ of horizontal edges of $ P $, $ e_{M} $ is called \textit{local maximum} if $ e_M $ be higher than two neighbor horizontal edges ($ y(e_{M})> y(e_{M-1}) $ and $ y(e_{M})> y(e_{M+1}) $) and also, $ e_{m} $ is called \textit{local minimum} if $ e_m $ is lower than two neighbor horizontal edges ($ y(e_{m})< y(e_{m-1}) $ and $ y(e_{m})< y(e_{m+1}) $). If edge $ \epsilon_1 \in E_U $ be a local maximum, then the internal angles of its both endpoints are equal to $ \frac{\pi}{2} $ and if $ \epsilon_2 \in E_U $ be a local minimum, then the internal angles of its two endpoints are $ \frac{3\pi}{2} $. If edge $ \epsilon_3 \in E_L $ be a local minimum, then the internal angles of its both endpoints are equal to $ \frac{\pi}{2} $ and if $ \epsilon_4 $ be a local maximum, then the internal angles of its two endpoints are $ \frac{3\pi}{2} $. If $ e( u_{M}) $ be local maximum then $ u_{M} $ is called local maximum and if $ e( u_{m}) $ be local minimum then $ u_{m} $ is called local minimum. Every rectangle $ R_i $ has the height $ h_i=\abs{y(u_i)-y(l_i)} $ , so, in the set $ R $ of rectangles obtained by vertical decomposition, rectangle $ R_{l} $ is called \textit{local maximum} if its height be greater then two adjacent rectangles ($ h_l>h_{l-1} $ and $ h_l>h_{l+1} $), and $ R_y $ is named \textit{local minimum} if its height be less than two adjacent rectangles ($ h_y<h_{y-1} $ and $ h_y<h_{y+1} $). Every rectangle has two adjacent rectangles except the first and last ones which are have only one adjacent rectangle. Two objects $ o $ and $ o' $ in polygon $ P $ are defined as \textit{weak visible} if every point of $ o $ is visible to some point of $ o' $. The interior area of polygon $ P $, as denoted $int(P)$, is the set of points that are bounded by $ P $, the exterior area of $ P $, as denoted $ext(P)$, is the set of the nearby and far away exterior points and the boundary of $ P $, as denoted $bound(P)$, is the set of all points on the boundary of $ P $. Clearly, a polygon is union of $ int(P)$ and $ bound(P) $. If $ e $ be a horizontal line segment, then left endpoint of $ e $ is denoted as $ left(e) $ and right endpoints of $ e $ is denoted as $ right(e) $. , also, if $ e $ be vertical, then top endpoint of $ e $ is denoted as $ top(e) $ and down endpoint of $ e $ is denoted as $ down(e) $. A \textit{star-shape} polygon is a polygon $ \rho $ that has some internal points so that the entire $ \rho $ is straight-line visible form each of them. Similarly, A \textit{$ r $-star} polygon is an orthogonal polygon $ \varrho $ such that there exist some internal points which the entire $ \varrho $ is orthogonally visible ($ r $-visible) form each of them, the set of these internal points that are visible from the entire polygon is named \textit{kernel}. If a polygon has a kernel, we are able to cover it with only one guard. Therefore, the problem of the decomposition an orthogonal polygon to the minimum number of $ r $-star sub-polygons is equal to the problem of guarding an orthogonal polygon with the minimum number of $ r $-guards, that is claimed in the previous section. The \textit{bounding box} of a set of objects is the minimum area box (rectangle) within which all the objects lie. So, the bounding box of a polygon $ P $ is the axis-aligned minimum area rectangle within which all the points of $ P $ lie, for the orthogonal polygons, their bounding boxes are edge aligned, too. If an $ x $-monotone ($ y $-monotone) orthogonal polygon has a horizontal (vertical) edge in common with its bounding box(rectangle), the polygon is called \textit{histogram}, the common edge is named \textit{base}. If an orthoconvex polygon has an edge in common with its bounding rectangle, the polygon is named \textit{pyramid}. Every pyramid polygon is histogram, too. An orthoconvex polygon that has two adjacent edges in common with its bounding box is called \textit{fan} polygon. Fan polygons are also pyramid and histogram and at least one of their vertices belongs to their kernel. Base on the presented classification in paper~\cite{palios2014minimum}, clearly, histogram belong to class-1 of orthogonal polygons, while $ x $-monotone(or $ y $-monotone) belong to class-2. In the class-2 of orthogonal polygons, members have dent edges in two different directions, for monotone polygons these two directions are parallel. This subclass of class-2 is denoted as \textit{class-2(a)} and if the two directions are perpendicular, the subclass is denoted as \textit{class-2(b)}. In the following, we prove the adapted lemma~\ref{le:lemma1} that was originally presented in~\cite{hoorfar2017linear}.
\begin{figure}
	\centering
	\includegraphics[width=0.8\textwidth]{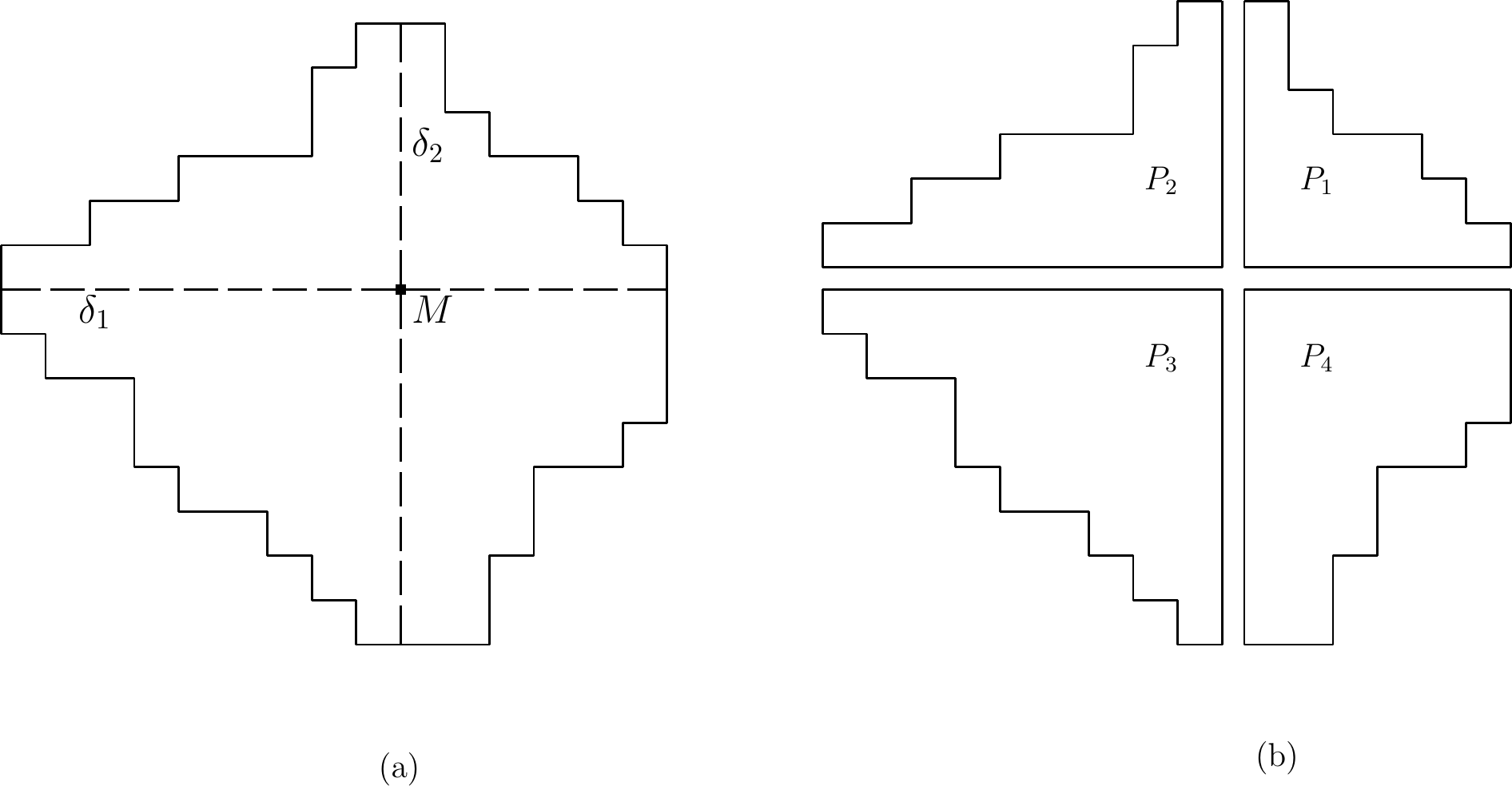}
	\caption{(a)A $ r $-star polygon $ P $ that has $ \delta_{1} $ and $ \delta_{2} $. Point $ M $ is the intersection of $ \delta_{1} $ and $ \delta_{2} $. (b)The decomposition of polygon $ P $ into four parts $ P_{1}$,$ P_{2}$,$ P_{3}$ and $ P_{4}$ which are fan polygons}
	\label{fi:fig2}
\end{figure}
\begin{lemma}
\label{le:lemma1}
An orthogonally convex (orthoconvex) polygon $ P $ is $ r $-star, if the leftmost and rightmost vertical edges of $ P $ are mutually weak visible and the upper and lower horizontal edges of $ P $ are mutually weak visible, too.
\end{lemma}
\begin{proof}
Because the leftmost and rightmost vertical edges of $ P $ are mutually weak visible, there exists a horizontal line segment $ \delta_{1} $  which is connecting the leftmost and rightmost vertical edges of $ P $ such that lies in $ P $. Similarly, Because the upper and lower horizontal edges of $ P $ are mutually weak visible, there exists a vertical line segment $ \delta_{2} $  which is connecting the upper and the lower horizontal edges of $ P $ such that lies in $ P $. If $ \delta_{1} $ connects the leftmost and rightmost vertical edges of $ P $ and $ \delta_{2} $  connects the upper and lower horizontal edges of $ P $ then they have an intersection $M$ that is contained in $ P $. $ \delta_{1} $ and $ \delta_{2} $ divide $P$ into 4 parts $ P_{1}$,$ P_{2}$,$ P_{3}$ and $ P_{4}$. All obtained pars are fan polygons with $M$ as their common core vertex. In every part, the entire $ M $ it is in the kernel, hence, if guard $ g $ is placed in the kernel, every point in $ P $ is visible to it.
\end{proof}
In the next section, we present a linear-time exact algorithm for finding the minimum guarding of orthogonal path polygons. Our algorithm uses the geometric approach that is presented in our previous researches~\cite{hoorfar2017minimum,hoorfar2017linear} to improved and obtain new results for orthogonal art gallery problem. Using this geometric approach instead of current graph theoretical leads to the algorithms with improved and better time complexity. In this approach, we  find the exact geometric positions of the point guards. Therefore, some of our definitions and notations is similar to our cited papers.

\section{An Algorithm for Guarding Path Galleries}
The path polygon has this property which can be divided into a number of sub-polygons, each of which can covered independently. The shortest watchman route of these sub-polygons is an orthogonal straight line segment. Also, for every sub-polygon we will prove that there is an optimum guard set which is all its guards are placed on the shortest watchman route of the sub-polygon. Hence, we will find that optimum guard set that is located on a set of line segments and it reduces the execution time of the algorithm. Besides that we will show that the visibility areas of all the guards that are located in a sub-polygon have not any effective intersections with the visibility areas of the guards that are located in another sub-polygons. So, the minimum number of guards that are required for guarding path polygon will be equal to the sum of the minimum numbers of guards that are required for the obtained sub-polygons. These sub-polygons are named \textit{balanced} orthogonal polygon that are monotone and \textit{straight-line walkable}. Straight-line walkable polygon means a polygon that its shortest watchman route(path) is a line segment i.e. a mobile guard can cover the entire polygon by walking back and forth on a straight route. For orthogonal polygon, this concept corresponds to the concept of balanced polygon. Actually, the described polygons have a area as named \textit{corridor} that straight shortest path is a part of it. For example, consider a histogram polygon, its base edge is a watchman route that is a part of its corridor. We will find this corridor using a ray-shooting (beam throwing) method in the next subsection. A path polygon is not necessarily straight walkable (or balanced), therefore, we will decompose a path polygon into the minimum number of balanced parts, then locating guards for every part, separately. At the first, the path polygon belongs to class-$ 4 $ of the described orthogonal classification, but after this decomposition all the obtained parts are belong to class-$ 2 $, because all the dent edges of path polygon that have horizontal direction are removed after the partition.
 
\subsection{The Decomposition of a Path Polygon into the Balanced Parts}
\begin{figure}
	\centering
	\includegraphics[width=\textwidth]{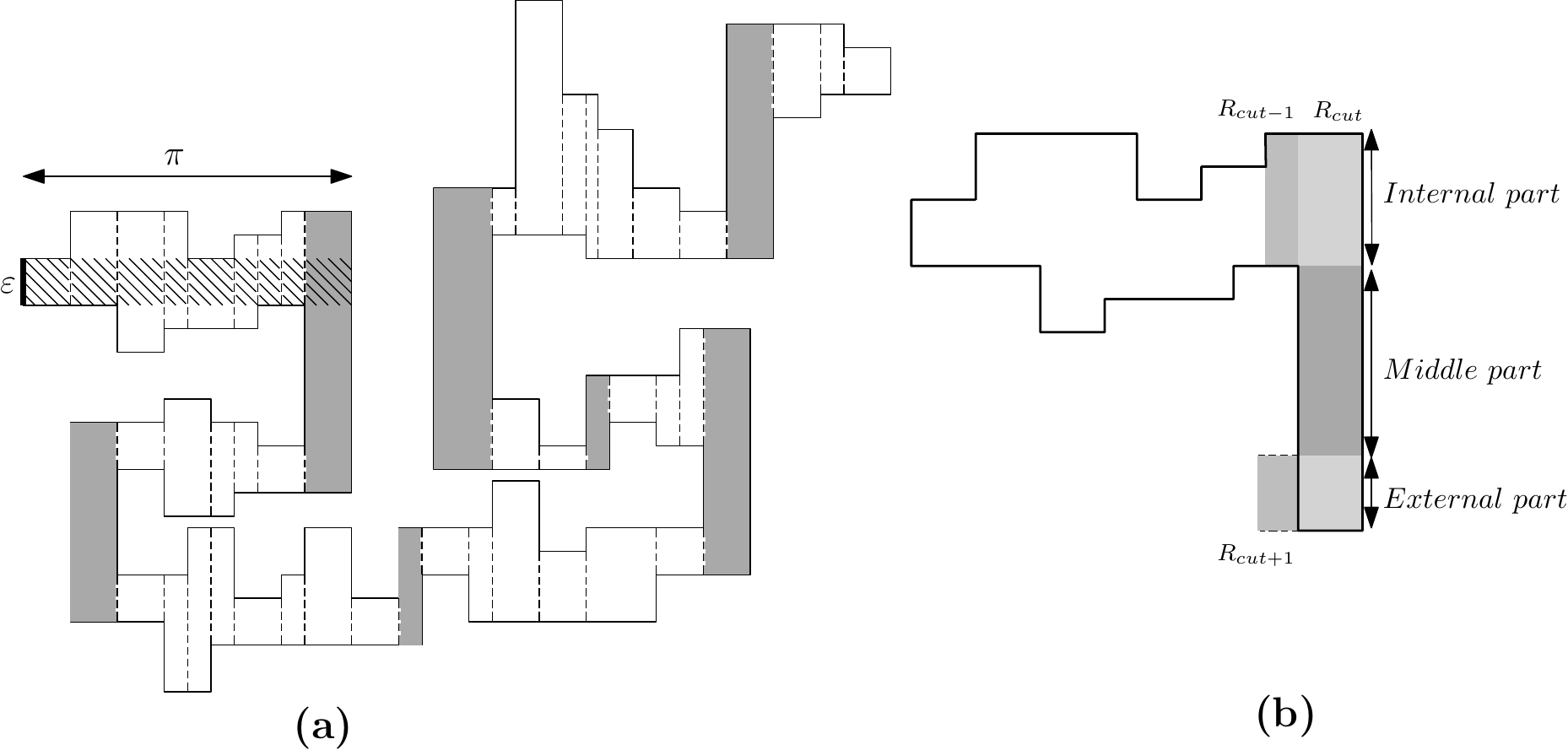}
	\caption{(a)Decomposition of a path polygon into balanced parts and vertical decomposition of them. The rectangles are shown in dark gray are cut. (b) An obtained balanced polygon $\pi$ and its corridor.}
	\label{fi:fig3}
\end{figure}
Suppose $ P $ be path polygon with $ n $ vertices that set $ R $ be its rectangle parts that are obtained after vertical decomposition and $ U $ and $ L $ be the sets of upper and lower edges of these obtained rectangles corresponding to the definitions that is explained in the previous section. Two of these rectangles are sources which are have only one adjacent parts while another have exactly $ 2 $, one of source rectangle is considered as start and another as last, also corresponding to the described order. The start rectangle and the general path polygon have one vertical edge in common, as denoted $\varepsilon $. Propagate a light beam in rectilinear path perpendicular to $ \varepsilon $ and also collinear with the $ X $-axis. Whole the light beam or a part of it passes through some members of the set $ R $(name this subset $ R_{\pi} $) and these rectangles together make a sub-polygon $ \pi $ of $ P $. See figure~\ref{fi:fig3}(a). The rectangles that belong to the polygon $ \pi $ have this geometric property which $ y $-ordinates of their upper edges are greater than $ y $-ordinates of their lower edges i.e. in polygon(sub-polygon) $ \pi $ all the dent edges of upper chain are higher than all the dent edges of lower chain. it is established that $ \min_{u_{i} \in \pi} (y(u_{i}))\geq \max_{l_{j} \in \pi}(y(l_{j})) $ for every $ u_{i}$ and $l_{i} $ belongs to rectangles of $ R_{\pi} $. Hence, there is a rectangular corridor $ \varsigma $ which is connecting the leftmost and rightmost vertical edge of $ \pi $ so that $ \varsigma $ has no intersection with $ ext(\pi) $ and contained in $ \pi $. If the leftmost vertical edge of sub-polygon $ \pi $ is denoted as $ v $, the rightmost vertical edge is denoted as $ v' $, also let $ y_1=\min_{u_{i} \in \pi} (y(u_{i})) $ and $ y_2=\max_{l_{j} \in \pi}(y(l_{j})) $, then $ \varsigma $ is a axis-aligned rectangle spanned by two points with the coordinates $ (x(v), y_1) $ and $ (x(v'), y_2)$. Therefore, $ \pi $ is walkable and balanced and every horizontal line segment that is connecting $ v $ and $ v' $ and located in $ \varsigma $ can be its shortest watchman route. After recognizing the first balanced sub-polygon(part) $ \pi $, we remove it from the path polygon $ P $ and iterate these operations to find next balanced parts until $ P $ is decomposed completely into balanced and monotone parts. \paragraph{Only one important point remains to be cleared up.} It is about the last rectangles of every obtained balanced sub-polygons(parts). suppose that $ P $ is decomposed into the balanced parts (polygons) $ \pi_{1},\pi_{2},\dots,\pi_{k} $ and let call the set of rectangles that is located in $ \pi_i $, $ R_{\pi_i} $, for every integer $ 1\leq i\leq k $. The last rectangle in $ R_{\pi_i} $ is called \textit{cut} rectangle because the intersection between $ \pi_i $ and $ \pi_{i+1} $ is the left edge of the cut rectangle(as denoted $ R_{cut} $). In fact, cut rectangle $ R_{\pi_i} $ is a border area and can either belong to the current part $ \pi_i $ or next part $ \pi_{i+1} $. If we want to find the path polygon $ P $ with the minimum number of guards, there \textbf{may} be a difference between two cases, that $ R_{cut} $ belongs to $ \pi_i $ {\small (case $ 1 $)} or $ R_{cut} $ belongs to $ \pi_{i+1} $ {\small (case $ 2 $)}. Let name previous rectangle of $ R_{cut} $, $ R_{cut-1} $, it is optimum that if $ R_{cut-1} $ is a local minimum, then we assign $ R_{cut} $ to $ \pi_{i+1} $(case 2). We prove this proposition in \textbf{lemma~\ref{le:lemma2}}. 
\paragraph{} Every cut rectangle is divided into three disjoint parts obtained by extending the horizontal edges of $ R_{cut-1} $ and $ R_{cut+1} $ incident to their common vertices with $ R_{cut} $ until intersect the boundary. The parts are called as \textit{internal}, \textit{middle} and \textit{external} parts, for an illustration see figure~\ref{fi:fig3}(b). The part that adjacent to $ R_{cut-1} $ is called \textit{internal part}, and the part that adjacent to $ R_{cut+1} $ is called \textit{external part} and third one is called \textit{middle part}. It is necessary to place a guard in cut rectangle for covering it, because it is impossible that the interior of middle part be guarded with an $ r $-guard that is not located in the cut rectangle $ R_{cut} $. If the previous rectangle $ R_{cut-1} $ be a local minimum, then we delete the cut rectangle from the set $ R_{\pi_i} $ and allocate it to the set  $ R_{\pi_{i+1}} $. Using this strategy reduces the number of required guards in some cases. For simplicity, we claim that:
\begin{claim}
\label{cl:claim1}
There exists a minimum cardinality guard set $ G={g_1, g_2,g_3\dots,g_{opt}} $ for a path polygon $ P $ so that all guards are located in the corridors.
\end{claim}
In these paper, we want to find the guard set $ G $ for path polygon $ P $ that is optimum and all its guards are located in corridors $ \varsigma_{1},\varsigma_{2},\dots,\varsigma_{k} $ of the obtained balanced sub-polygons $ \pi_{1},\pi_{2},\dots,\pi_{k} $ that $ k $ is the minimum number of sub-polygons($ 1\leq k \leq \lfloor\frac{n}{4}\rfloor $).
\begin{lemma}
\label{le:lemma2}
It is the optimum for guarding path polygon $ P $ that if $ R_{cut-1} $ is a local minimum, then we assign $ R_{cut} $ to $ \pi_{i+1} $ instead of assigning it to $ \pi_i $.
\end{lemma}
\begin{proof}
\begin{figure}
	\centering
	\includegraphics[width=\textwidth]{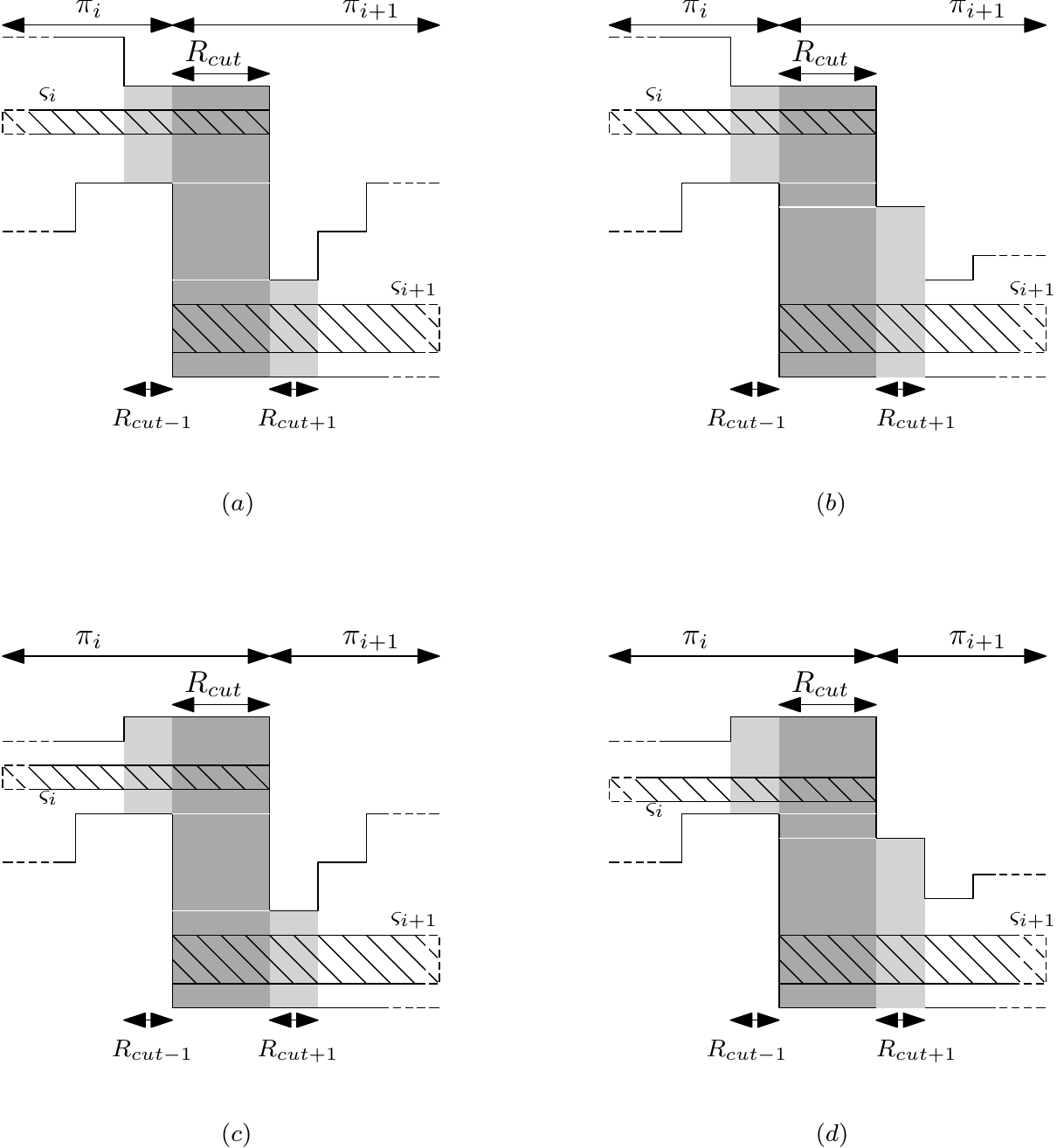}
	\caption{Four different cases occur for assigning the cut rectangle $ R_{cut} $ to $ \pi_i $ or $ \pi_{i+1} $.}
	\label{fi:fig4}
\end{figure}
Suppose that after the decomposition of the path polygon $ P $, for an integer $ i $, $ \pi_i $ and $ \pi_{i+1} $ be two adjacent sub-polygons and a cut rectangle $ R_{cut} $ is located between them, as shown in figures~\ref{fi:fig4}. Let $ \varsigma_i $ and $ \varsigma_{i+1} $ are corridors of $ \pi_i $ {\footnotesize (or $ \pi_i \cup R_{cut} $)} and $ \pi_{i+1} $ {\footnotesize (or $ \pi_{i+1} \cup R_{cut} $)}, respectively. The previous rectangle of $ R_{cut} $ is called $ R_{cut-1} $ and the next rectangle of $ R_{cut} $ is called $ R_{cut+1} $.  As already mentioned, because the interior of middle part of cut rectangle $ R_{cut} $ is not orthogonally visible from any points of $ P-R_{cut} $, it is necessary to place a guard $ g $ in $ R_{cut} $ for guarding it. Where is the best position for this guard $ g $? locating $ g $ in the intersection between $ R_{cut} $ and $ \varsigma_i $ (as denoted $ R_{cut}\cap \varsigma_i $) or in the intersection between $ R_{cut} $ and $ \varsigma_{i+1} $ (as denoted $ R_{cut}\cap \varsigma_{i+1} $) is better than anywhere else in $ R_{cut} $. If we locate $ g $ in $ R_{cut}\cap \varsigma_i $, $ g $ guard $ R_{cut} $ and some rectangles before it which belong to $ \pi_i $ and if we locate $ g $ in $ R_{cut}\cap \varsigma_{i+1} $, $ g $ guard $ R_{cut} $ and some rectangles after it which belong to $ \pi_{i+1} $. Which one lead to the minimum guarding of path polygon $ P $? There are four different cases which are shown in figure~\ref{fi:fig4}. In two cases (a) and (b), $ R_{cut-1} $ is a local minimum {\footnotesize (because $ h_{cut-1}<h_{cut} $ and $ h_{cut-1}<h_{cut-2} $)}, by placing $ g $ in the area $ h_{cut} \cap \varsigma_i $ the rectangle $ R_{cut-1} $ is guarded but the rectangle $ R_{cut-2} $ is not guarded (completely). So, certainly, there is a guard $ g' $ in the guard set that cover (guard) rectangle $R_{cut-2} $ and we know that the height of $ R_{cut-2} $ is higher than the height of $ R_{cut-1} $, hence $ g' $ can also guard $ R_{cut-1} $ completely. So, if the rectangle $ R_{cut-1} $ be local minimum, then locating $ g $ in the area $ h_{cut} \cap \varsigma_i $ is not useful. Therefore, it is better to locate $ g $ in the area $ h_{cut} \cap \varsigma_{i+1} $. It happens if we assign the cut rectangle $ R_{cut} $ to $ \pi_{i+1} $ instead of assigning it to $ \pi_i $ {\footnotesize (whether $ R_{cut+1} $ be local minimum or not)}. In the case (c), $ R_{cut+1} $ is a local minimum {\footnotesize (because $ h_{cut+1}<h_{cut} $ and $ h_{cut+1}<h_{cut+2} $)}, by placing $ g $ in the area $ h_{cut} \cap \varsigma_{i+1} $ the rectangles $ R_{cut} $ and $ R_{cut+1} $ are guarded but the rectangle $ R_{cut+2} $ is not guarded (completely). So, certainly, there is a guard $ g' $ in the guard set that cover (guard) rectangle $R_{cut+2} $ and we know that the height of $ R_{cut+2} $ is higher than the height of $ R_{cut+1} $, hence $ g' $ can also guard $ R_{cut+1} $ completely($ g' $ is located somewhere in $ \varsigma_{i+1} $). So, if the rectangle $ R_{cut+1} $ be local minimum, then locating $ g $ in the area $ h_{cut} \cap \varsigma_{i+1} $ is not useful. Therefore, it is better to locate $ g $ in the area $ h_{cut} \cap \varsigma_{i} $. It happens if we assign the cut rectangle $ R_{cut} $ to $ \pi_i $ instead of assigning it to $ \pi_{i+1} $, while $ R_{cut-1} $ is not local minimum. In the case (d), both of $ R_{cut-1} $ and $ R_{cut+1} $ are not local minimum. The rectangle $ R_{cut-1} $ is not local minimum and the height of $ R_{cut-1} $ is higher than the height of $ R_{cut-2} $, so, for guarding $ R_{cut-1} $ it is necessary to place a guard in area $ (R_{cut-1}\cup R_{cut})\cap \varsigma_i $. Also, the rectangle $ R_{cut+1} $ is not local minimum and the height of $ R_{cut+1} $ is higher than the height of $ R_{cut+2} $, so, for guarding $ R_{cut+1} $ it is necessary to place a guard in area $ (R_{cut}\cup R_{cut+1})\cap \varsigma_{i+1} $. Well, we do not need two guards in $ R_{cut} $ then, only for simplicity, we locate one guard in area $ R_{cut}\cap \varsigma_i $ and another guard in area $ R_{cut+1}\cap \varsigma_{i+1} $. It happens when we assign the cut rectangle $ R_{cut} $ to $ \pi_i $ instead of assigning it to $ \pi_{i+1} $. 
\begin{figure}
	\centering
	\includegraphics[width=0.8\textwidth]{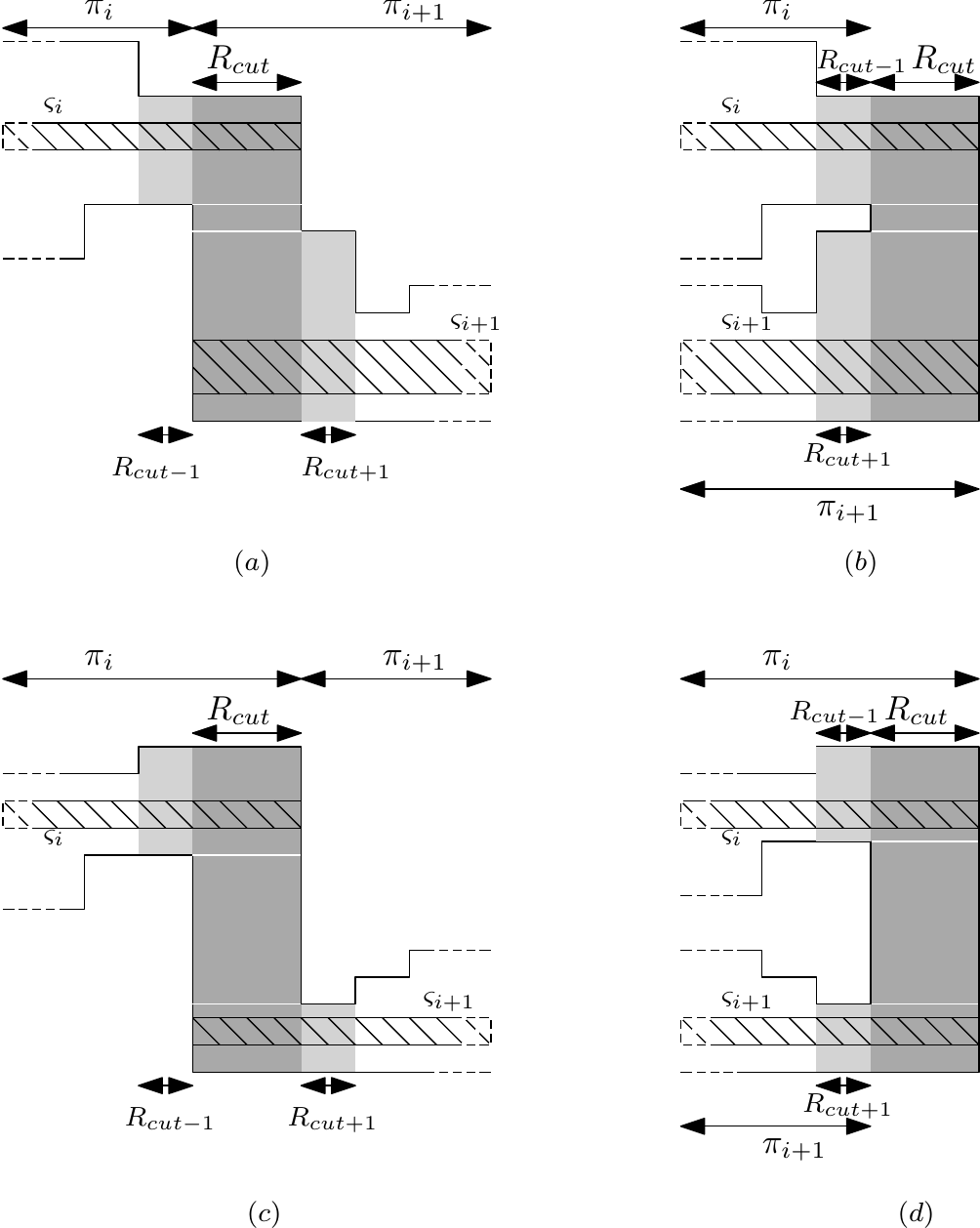}
	\caption{An illustration of the cases occur when two adjacent rectangles of $ R_{cut} $ are located on the same side.}
	\label{fi:fig5}
\end{figure}
In four described cases that are shown in figure~\ref{fi:fig4}, two adjacent rectangles $ R_{cut-1} $ and $ R_{cut+1} $ are located on different sides of $ R_{cut} $. Clearly, if the rectangles $ R_{cut-1} $ and $ R_{cut+1} $ are located in the same side, four another cases are occurred that they are similar to the previous four cases. For an illumination see figure~\ref{fi:fig5}, the case that is shown in (b) is equal to (a) and the case that is shown in (d) is equal to (b). Therefore, we do not focus on these new four cases.
\end{proof}
Remember the decomposition of path polygon $ P $ into the balanced sub-polygons and suppose that we find first balanced sub-polygon of $ P $, so, we remove it from $ P $ and iterate algorithm for $ P-\pi $ until $ P $ is decomposed into several balanced $ x $-monotone polygon. We remove the rectangles belong to $ \pi $ belong to $ R $. We know the members of $ R $ are ordered and labeled from $ 1 $, after removing, we relabel the remained members from $ 1 $, again, to simplify the description of the algorithm. Certainly, the same processes will be occurred for $ U $ and $ L $. The number of iterations is equal to the cardinality of $ R $ (in the beginning). Therefore, the time complexity of the decomposition path polygon $ P $ into balanced sub-polygons is processable in the linear-time corresponding to the size of $ P $. Now, we describe the linear-time algorithm for decomposition $ P $ into the balanced sub-polygons.
\begin{algorithm}[]
	\KwData{an path polygon with $ n $ vertices}
	\KwResult{minimum number of balaced monotone polygons}
	set $ min_u=u_1$ and $ max_l=l_1 $\;
	\While{set of rectangles $ R \neq \emptyset $}{
		\eIf{ $ u_i>max_l $ or $ l_i<min_u $}{
			\eIf{$ i-2 = 1 $ or $ R_{i-2} $ is not local minimum}{
			   $ R = R - \{R_1,R_2,\dots,R_{i-2},R_{i-1} \}$ \;
			   $ U = U - \{u_1,u_2,\dots,u_{i-2},u_{i-1} \}$ \;
			   $ L = L - \{l_1,l_2,\dots,l_{i-2},l_{i-1} \}$ \;}{
			   $ R = R - \{R_1,R_2,\dots,R_{i-2}\} $\;
			   $ U = U - \{u_1,u_2,\dots,u_{i-2}\} $\;
			   $ L = L - \{l_1,l_2,\dots,l_{i-2}\} $\;		   
			   }
			refresh the index of  $ R $, $ U $ and $ L $ starting with $ 1 $\;
			reset $ min_u=u_1$ and $ max_l=l_1 $\;
		}{
		
			set $min_u=\min(min_u,u_{i})$ and $max_l=\max(max_l,l_{i})$\;
		}
	}
\caption{The algorithm for decomposition path polygon $ P $ into the balanced sub-polygons.}
\label{al:algo1}
\end{algorithm}  
Every balanced and monotone polygon $ \pi $ has an axis-aligned rectangular area $ \varsigma $ (named corridor) which is connecting the leftmost and rightmost edges of $ \pi $. This area is also connecting the lowest dent edge of upper chain and the highest dent edge of lower chain. Suppose that $ P $ is decomposed into a set of the balanced sub-polygons $ \pi_{1},\pi_{2},\dots,\pi_{k} $ and $ \varsigma_{1},\varsigma_{2},\dots,\varsigma_{k} $ be their corridors, respectively. So, if $i \neq j$, There is no point in the interior of $ \varsigma_i $ such that orthogonally visible from $ \varsigma_j $. Due to this fact, if we optimally cover $ P $ so that all the guards are located on the corridors, guarding each of sub-polygons can be done independently i.e. the minimum number of guards for guarding the entire polygon is the sum of the minimum number of guards that are necessary for every sub-polygons. One addition point is about proving the explained claim. To prove claim~\ref{cl:claim1}, we present an algorithm in the next sections and prove that its results is optimum.

\subsection{The Algorithm for Guarding the Balanced Sub-polygons}
\label{ss:ss01}
In the previous subsection, we explained that every balanced (walkable) polygon has a rectangular area, named corridor which is the entire polygon is weak visible from it. Now, we describe an algorithm to find the minimum number of guards and their positions for an orthogonal and monotone balanced polygon, such that all guards is only located in the corridor. The presented algorithm in this section is the improved version of the algorithm that is already presented in our paper~\cite{hoorfar2017minimum}. Assume that $ P $ is a balanced orthogonal monotone with $ n $ vertices, after vertical decomposition, the sets $ R $, $ U $, $ L $, $ E_L $ and $ E_U $ are obtained for the polygon $ P $ according to their definitions. Let $ \varepsilon $ and $ \varepsilon' $ be the leftmost and rightmost vertical edges of $ P $ and $ e_{min} $ and $ e_{max} $ be the lowest horizontal edge of the upper chain of $ P $ and the highest horizontal edge of the lower chain of polygon $ P $, respectively.
\begin{defini}
An axis-aligned rectangular area that is contained in $ P $ and spanned by points $ (x(\varepsilon),y(e_{min}))  $ and $ (x(\varepsilon'),y(e_{max})) $ is named corridor of $ P $. The corridor of a balanced monotone polygon $ P $ is not empty and denoted as $ \varsigma_p $.
\end{defini} 
\begin{defini}
For a horizontal edge $ e $ of the polygon $ P $, the set of every point $p\in P$ which there is a point $ q \in e $ such that $ pq $ is a line segment normal to $ e $ and completely inside $ P $, is named \textit{orthogonal shadow} of $ e $, as denoted $ os_e $(for abbreviation).
\end{defini}
For the balanced monotone orthogonal $ P $ with $ n $ vertices, we present algorithm~\ref{al:algo2} to find the minimum number of guards and their positions. In the following, we explain the details of the algorithm and illustrate it. First, we find all tooth edges of the set $ E = E_L \cup E_U $ and call the obtained set as $ D $. For every $  d_i \in D $, we compute orthogonal shadow of $ d_i $ as $ os_{i} $. Let $ D=\{d_1, d_2,\dots,d_k\} $ and $ OS=\{os_1,os_2,\dots,os_k\} $, ordered from left to right by $ x $-coordination of their left vertical edges.
\begin{lemma}
\label{le:lemma02}
Every tooth edge $ t $ can be covered just with a guard which is placed in orthogonal shadow $ os_{t} $, not anywhere else. 
\end{lemma}
\begin{proof}
Suppose that $ P $ is a monotone orthogonal polygon and assume that the tooth edge $ t $ is guarded with $ \gamma $ that is not placed in the shadow $ os_{t} $, so, $ \gamma $ is not in the $ x $-coordinate of any points on the edge $ t $. Assume that the left and right endpoints of $ t $ are denoted as $ L_t $ and $ R_t $, and let $ x $-coordinate of $ L_t $ be greater than $ x $-coordinate of $ \gamma $ i.e. $ x(L_t )> x(\gamma)$. Clearly, $ \gamma $ is visible to $ t $, so, two endpoints $ R_t $ and $ L_t $ are visible to $ \gamma $. Hence, there exists an axis-aligned rectangle spanned by the $ R_t $ and $ \gamma $ is contained in $ P $. These two points is not in the same $ x $-coordinate and even in the same $ y $-coordinate, we know that every vertices of the rectangle belong to $ P $ as denoted $ R_t=(x(R_t),y(R_t)) $, $ A=(x(R_t),y(\gamma)) $, $ B=(x(\gamma),y(R_t)) $ and $ g_d=(x(\gamma),y(\gamma))$. So, the horizontal edge $ BR_t $ is contained in $ P $, Completely. It is impossible, because $ t \subset BR_t $ i.e. if an edge of polygon be only a part of a segment which belong to the polygon, So, it is not really an edge.
\end{proof}
\begin{figure}
	\centering
	\includegraphics[width=\textwidth]{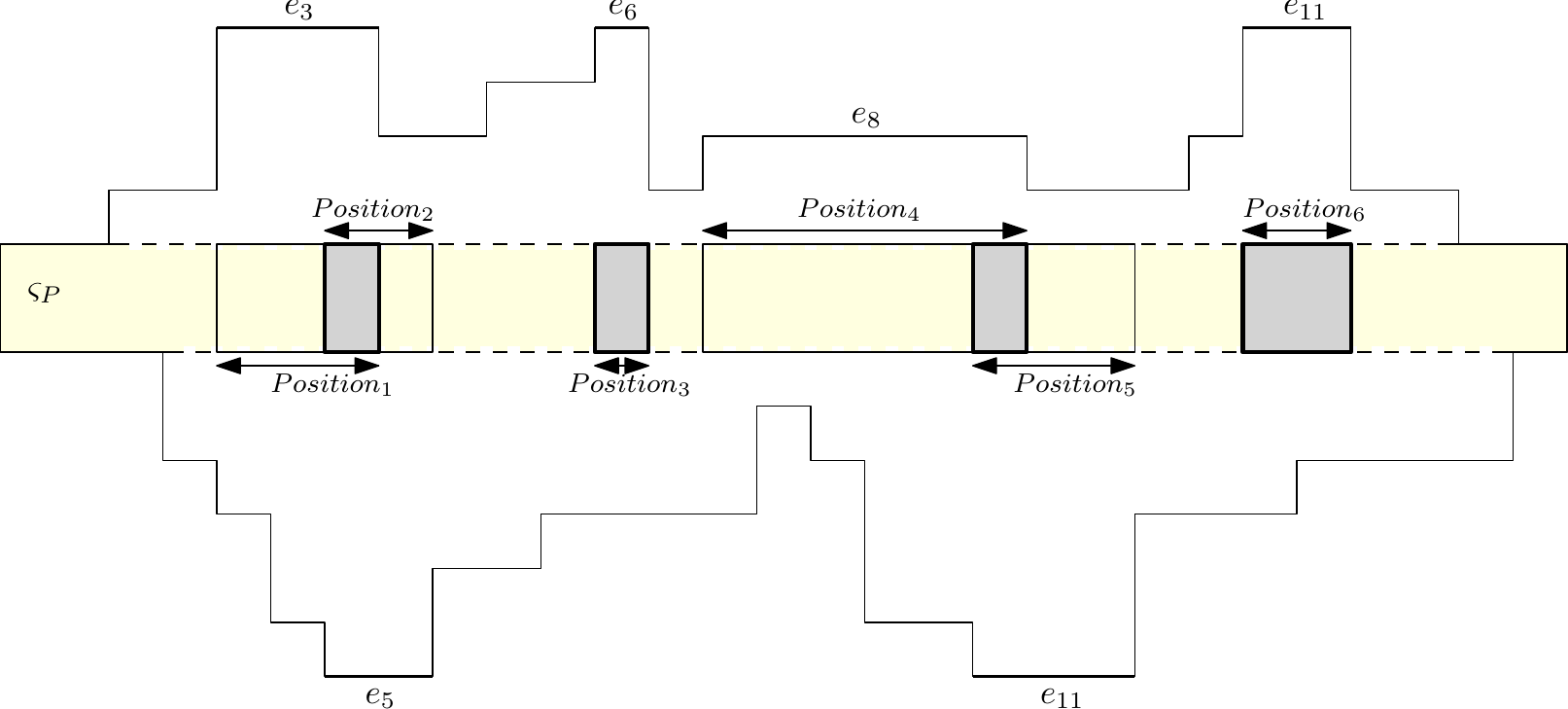}
	\caption{An Illumination of the definitions, the tooth edges are shown in bold segments and all bold bordered rectangles of $ \varsigma $ belong to $ Position $ that are the positions for guards.}
	\label{fi:fig6}
\end{figure}
Hence, according to the lemma~\ref{le:lemma02}, at least, one guard must be placed in every orthogonal shadow of tooth edges. We show in the algorithm that this number of guards is sufficient for guarding the entire polygon $ P $ and no extra guard is needed. See figure~\ref{fi:fig6}, some orthogonal shadows of tooth edges of the upper chain may have intersection with some orthogonal shadows of tooth edged of the lower chain. If it happen, we place a guard in the intersection between them to reduce the number of guards i.e. if two different tooth edges $ t $ and $ t' $ belong to $ E_U $ and $ E_L $, respectively, and the intersection between their orthogonal shadows is not empty (as denoted $ os_{e_1}\cap os_{e_2}\neq \emptyset $) for guarding both of them one guard on the intersection is sufficient. So, in the $ OS $, we replace two members $ os_{t}$ and $ os_{t'}$ with the intersection of them($ os_{e_1}\cap os_{e_2}$). We know that the intersection of every $ 3 $ members of $ OS $ is empty and after these replacement the cardinality of set $ OS $ is equal to $ \kappa \leq k $. Remember that the strategy of our guarding is placing guards in the corridor of balanced monotone orthogonal polygon $ P $ and we know the orthogonal shadow of every tooth edge of $ P $ has intersection with $ \varsigma $. Assume that the rectangular area $Position_i=os_i \cap \varsigma $ {\footnotesize (for every $ i $ between $ 1 $ and $ \kappa $)} and $Position=\{Position_1,Position_2,\dots,Position_\kappa\} $ s.t. $ (\kappa \leq k) $, ordered corresponding to their rectangle order. Now, we know that the intersection of every $ 2 $ elements of set $ Position $ is empty. The set $ Position $ is the positions for placing guards, one guard must be located in every member of $ Position $, see figure~\ref{fi:fig4} again. Using this strategy leads to find the positions for locating the minimum number of guards in the balanced orthogonal sub-polygon $ P $ in the linear time corresponding to number of its vertices $ n $. In algorithm~\ref{al:algo2}, the set $ Position $ is the positions for the optimum guard set and the variable $ GuardNumber $ is the cardinality of the optimum guard set.
\begin{algorithm}[]
	\KwData{the horizontal edges of two chains of balanced monotone orthogonal polygon $ P $ with $ n $ vertices ($ E_L $,$ E_U $)}
	\KwResult{the optimum number of point guards ($ GuardNumber $) and their positions ($ Position $)}
	Set $ GuardNumber=0 $ and $ Position=\emptyset $\;
	Set $ e_{min}= the~lowest~horizontal~edge~of~E_U $\;
	Set $ e_{max}= the~highest~horizontal~edge~of~E_L $\;
	\ForEach{edge $ e_{i} $ belongs to $ E_L $}{
		\If{Interior angles of $right(e_i)$ and $ left(e_i)$ are equal to $\frac{\pi}{2}$}{
			$ A_i=(x(left(e_i)),y(e_{max}) $\;
			$ B_i=(x(right(e_i)),y(e_{min}) $\;
			Set $ Position_i= $ rectangle spanned by  $ A_i $ and $B_i $\;
			Set $ Position_L=Position_L \cup \{Position_i $\}\;
			$ GuardNumber++ $;
		}
	}
   \ForEach{edge $ e_{i} $ belongs to $ E_U $}{
		\If{Interior angles of $right(e_i)$ and $ left(e_i)$ are equal to $\frac{\pi}{2}$}{
			$ A_i=(x(left(e_i)),y(e_{max}) $\;
			$ B_i=(x(right(e_i)),y(e_{min}) $\;
			Set $ Position_i= $ rectangle spanned by  $ A_i $ and $B_i $\;
			Set $ Position_U=Position_U \cup \{Position_i $\}\;
			$ GuardNumber++ $;
		}
   }
   Merge the sorted lists $ Position_L$ and $ Position_U $ as sorted list $ Position $.\;
   \ForEach{horizontal segment $ position_{i} $ belongs to $ Position $}{
   		\If{$ Position_i\cap Position_{i+1}\neq \emptyset $}{
   			$ Position_i=Position_i\cap Position_{i+1} $\;
   			$ Position=Position-\{Position_{i+1}\}$\;
   			$ GuardNumber-- $;
   		}
      }
\caption{Optimum guarding of a balanced monotone orthogonal polygon $ P $ with $ n $ vertices.}
\label{al:algo2}
\end{algorithm}
\\The positions of all guards are in the set $ SI $ and every elements of $ SI $ is a subset of corridor $ \varsigma $, so, all guards are located on corridor $ \varsigma $. It is clear that the time complexity of the algorithm is as same as the cardinality of set $ E $ and it is linear-time according to the size of $ E $.
\begin{lemma}
\label{le:lemma3}
The minimum number of guards for cover a balanced monotone orthogonal polygon $ P $ is equal to $ GuardNumber $ that is obtained by algorithm~\ref{al:algo2} .
\end{lemma}
\begin{proof}
Suppose that $ GuardNumber $ guards is sufficient to guard the entire polygon $ P $, using lemma~\ref{le:lemma02} prove that this number of guards necessary even for guarding the tooth edges of $ P $. Every area $ Position_i\in Position $ is a subset of a $ r $-star sub-polygon i.e. if we decompose $ P $ into $ r $-star parts(sub-polygons) then the kernels of every $ r $-star sub-polygons has at least one point in the elements of $ Position $, so the entire $ P $ is covered by these $ GuardNumber $ guards and their positions.
\end{proof}

\subsection{Time Complexity of Algorithm}
Now, we discuss about efficiency and time complexity of the whole solution and we explain that why our algorithm is processable in $ O(n) $- time while n be the size of the input(path polygon $ P $). Given path polygon $ P $, for guarding $ P $, we need to decompose the polygon into balanced parts with algorithm~\ref{al:algo1}. The vertical decomposition and finding optimum balanced orthogonal parts are solvable in the linear-time ($ O(n)-time $) because the number of rectangles is order of $ O(n)$. After that the problem is divided into subproblems which are finding minimum guard set for the obtained balanced sub-polygons(parts). We use algorithm~\ref{al:algo2} for guarding monotone parts, hence, subproblems is solvable in the linear-time corresponding to its size. The total time of solving sub-problems is $ O(n) $-time i.e. the total number of vertices of the all obtained balanced parts is $ O(n) $, so, algorithm~\ref{al:algo2} is run in $ O(n) $-time for all balanced parts. Therefore, all computations handle in $ O(n) $-time. Finally, $ GuardNumber $ is referred to the optimum number of guards needed to cover path polygon $ P $. Therefore, we have proved the general result of the paper:
\begin{theorem}
	\label{th:th01}
	There is a geometric algorithm that can find the minimum number of guards for given orthogonal path polygon $ P $ with $ n $ vertices, with r-guards in $ O(n) $-time.
\end{theorem}

\section{Conclusion}
We studied the problem of finding the minimum number of r-guards for an orthogonal path polygon. This problem is a well-known version that is named \textit{orthogonal art gallery problem}. The total target in the orthogonal art gallery problem is finding the optimum set of r-guards $ G $ which is a set of point guards in polygon $ P $ that all points of the $ P $ are orthogonally visible from at least one r-guard in $ G $. We present an exact optimum algorithm for finding the guard set for path galleries. We solved this problem in the linear time according to $ n $ where $ n $ is the number of sides of path polygon. the space complexity of our algorithm is $ O(n) $, too. Many of the algorithms presented in this field are based on graph theory, but our proposed algorithm is based on geometric approach which is presented in paper~\cite{hoorfar2017minimum}. This approach can lead to improved performance and efficiency in the algorithms. We use our previous strategy~\cite{hoorfar2017minimum} that was provided a purely geometric algorithm for the orthogonal art gallery problem where the galleries are monotone and extending the algorithm for path galleries. Actually, we improved the time complexity of the orthogonal art gallery problem for path polygons from $ O(n^{17} poly\log n) $-time~\cite{worman2007polygon} to linear-time. For the future works, we want to try to solve this problem for every simple orthogonal polygon with/without holes. Both time and space complexity of our presented algorithm is order of $ O(n) $ and it is the best for these galleries.

\bibliographystyle{plain}
\bibliography{bibfile}

\end{document}